\documentclass[runningheads]{llncs}

\usepackage{graphicx}
\usepackage{tikz}
\usetikzlibrary{automata,positioning}
\usepackage{amsmath}
\usepackage{amssymb}
\usepackage{ascmac}
\usepackage{mathbbol}
\usepackage[bb=boondox]{mathalfa}
\usepackage{comment}
\usepackage{physics}
\usepackage[all]{xy}
\usepackage{enumerate}
\usepackage{caption}

\renewcommand{\qed}{\hfill$\square$}
\newcommand{\bqed}{\hfill$\blacksquare$}
\newcommand{\TK}{\mathcal{T}_K}
\newcommand{\Can}{{\bf Can}}

\newenvironment{proof*}[1]
  {\renewcommand{\proofname}{#1}%
   \begin{proof}}
  {\end{proof}}

\begin{document}

\title{Semidirect Product Decompositions for \\Periodic Regular Languages}
\titlerunning{Semidirect Product Decompositions for Periodic Regular Languages}
%
%
\author{Yusuke Inoue \and
Kenji Hashimoto\and
Hiroyuki Seki}
\authorrunning{Y. Inoue et al.}
%
\institute{Nagoya University}
\maketitle              

\begin{abstract}
  The definition of period in finite-state Markov chains can be extended to regular languages by considering the transitions of DFAs accepting them. For example, the language $(\Sigma\Sigma)^*$ has period two because the length of a recursion (cycle) in its DFA must be even.
This paper shows that the period of a regular language appears as a cyclic group within its syntactic monoid. Specifically, we show that a regular language has period $p$ if and only if its syntactic monoid is isomorphic to a submonoid of a semidirect product between a specific finite monoid and the cyclic group of order $p$.
Moreover, we explore the relation between the structure of Markov chains and our result, and apply this relation to the theory of probabilities of languages. We also discuss the Krohn-Rhodes decomposition of finite semigroups, which is strongly linked to our methods.
\end{abstract}

\section{Introduction}
Numerous algebraic approaches to formal languages have been conducted.
For example, the algebraic characterization of star-free languages by Sch\"{u}tzenberger \cite{Schutzenberger} is one of the most famous results of such studies. Most of algebraic approaches to regular languages over an alphabet $\Sigma$ are based on the fact that a language $L\subseteq\Sigma^*$ is regular iff there exist a finite monoid $M$, a monoid homomorphism $\eta:\Sigma^*\to M$ and a subset $S\subseteq M$ such that $L=\eta^{-1}(S)$.
The smallest monoid $M$ that satisfies this condition is called the syntactic monoid of $L$, and the corresponding $\eta$ is called the syntactic morphism of $L$.
The relation between a regular language $L$ and its syntactic monoid $M$ is described by a Cayley graph as follows.
The Cayley graph of the syntactic monoid $M$ is the directed labeled graph such that the set of vertices is $M$ and the set of edges consists of $m_1\xrightarrow{a}m_2$ where $m_1\cdot \eta(a)=m_2$ with $m_1,m_2\in M$ and $a\in \Sigma$. Then, the Cayley graph can be regarded as the DFA $\mathcal{A}$ recognizing $L$ defined as:
$\mathcal{A}$ accepts $w\in \Sigma^*$ iff there is a path $e_M\xrightarrow{w}m$ where $e_M$ is the identity element of $M$ and $m$ is contained in the image $\eta(L)$.

In this paper, we focus on periods of regular languages. We say that a regular language $L\subseteq \Sigma^*$ has a period $P\ge 1$ with respect to $\Gamma\subseteq \Sigma$ if the Cayley graph of the syntactic monoid of $L$ satisfies: if $m\xrightarrow{w}m$ is a path with some $m\in M$, then the number of occurrences of letters in $\Gamma$ in $w$ is a multiple of $P$.
{(The reason for focusing on Cayley graphs will be explained in Section~2.3.)}
For example, the syntactic monoid of $(\Sigma\Sigma)^*$ is the cyclic group $C_2=\{0,1\}$ and the set of edges of the Cayley graph is $\{0\xrightarrow{a}1,1\xrightarrow{a}0\mid a\in\Sigma\}$.
Therefore, $(\Sigma\Sigma)^*$ has period $2$ with respect to $\Sigma$ because $w$ must be of even length if $m\xrightarrow{w}m$ is a path for each $m\in\{0,1\}$ and $w\in\Sigma^*$.
Note that a period of a language $L\subseteq\Sigma^*$ is defined for any given non-empty subset $\Gamma\subseteq \Sigma$.
As another example, we consider the language $L_1=\{w\mid$ the numbers of occurrences of $a$ and $b$ in $w$ are both even$\}\subseteq \{a,b\}^*$.
Then, the syntactic monoid of $L_1$ is the direct product $C_2\times C_2$ of two cyclic groups $C_2$. One $C_2$ counts the occurrences of $a$ and the other $C_2$ counts the occurrences of $b$.
Therefore, $L_1$ has period $2$ with respect to both of $\Gamma_1=\{a\}$ and $\Gamma_2=\{b\}$.

Now, how are these periods represented by syntactic monoids?
In the examples mentioned in the last paragraph, each period is simply represented by the corresponding syntactic monoid.
In the first example, the period of $(\Sigma\Sigma)^*$ with respect to $\Sigma$ is $2$, and it is represented by the cyclic group $C_2$.
In the second example, the periods of $L_1$ with respect to $\Gamma_1$ and $\Gamma_2$ are represented by two cyclic groups $C_2$, respectively.
However, a general case is not always as simple as these examples.
As we will see in later sections, there is a language that has period $2$, and its syntactic monoid is the symmetry group $\mathcal{S}_3$, which consists of all permutations on $\{0,1,2\}$.
Of course, $\mathcal{S}_3$ is not isomorphic to cyclic groups or their direct products. How does the syntactic monoid explain the period in such a case?

{
One of our goals is to provide an algebraic decomposition of the syntactic monoid of a periodic language.
As the main result of this paper, we show that every syntactic monoid of a regular language with periods $P_1,\ldots,P_n>1$ with respect to $\Gamma_1,\ldots,\Gamma_n\subseteq \Sigma$ can be decomposed into a submonoid of the semidirect product $N^{G}\rtimes G$ where $N$ is a specific finite monoid and $G$ is the direct product of cyclic groups of orders $P_1,\ldots,P_n$.
The second component $G$ of the decomposition clearly contains the period $P_i$ as the $i$-th cyclic group for each $1\le i\le n$, and the first component $N^G$ simulates the behavior of each periodic class of the syntactic monoid with states of which number is smaller than its order.}
The contributions of our results and related studies are as follows.

\begin{enumerate}[(i)]
  \item Our decomposition explains an iteration property of regular languages.
  Our definition of period is inspired by period in the context of Markov chains. For an irreducible Markov chain, its period is defined as the greatest common divisor of the lengths of all recursions (cycles) whose probabilities are positive. Period is a crucial concept in the analysis of Markov chains (see e.g., \cite{Norris}). For example, the existence of the limit distribution of a Markov chain depends on its period. We extend the concept of period in Markov chains to regular languages by considering the transitions of DFA accepting them. An essential idea of the extension is using syntactic monoids rather than minimal DFAs or other DFAs. Our study highlights several advantages of syntactic monoids in explaining a specific iteration property of regular languages.

  \item Periods of Markov chains are strongly related to the probability of regular languages, and therefore, our results have applications in the study of the probability of regular languages. The probability $\mu_L(\ell)$ of a language $L$ for length $\ell$ is the probability that $w$ belongs to $L$ when a word $w\in \Sigma^{\ell}$ is randomly chosen.
  For example, the probability of $L=a\Sigma^*$ where $\Sigma=\{a,b\}$ is $\tfrac{1}{2}$ for every length $\ell\ge 1$.
  Probabilities of languages have been studied in many different contexts. The most classical results on probabilities were obtained as an application of formal power series \cite{powerseries}.
  In recent years, there have been several approaches to probabilities using syntactic monoids (e.g., \cite{zero-one,LTmeasure}), and our study contributes to this body of work.

  \item Our decomposition theorem provides a partial Krohn-Rhodes decomposition of the transformation semigroup $(M_L,M_L)$ where $M_L$ is the syntactic monoid of a periodic regular language $L$.
  A transformation semigroup $(X, M)$ is the pair of a set $X$ and a semigroup $M$ that acts on $X$.
  The Krohn-Rhodes prime decomposition theorem \cite{Krohn-Rhodes} states that every finite transformation semigroup can be decomposed into a wreath product of finite {monoids having only trivial subgroups\footnote{Such monoids are usually said to be {\it aperiodic}, but we don't use this term to avoid confusion with period of a language, the key concept of this paper.}} and finite groups.
  This is a famous result in semigroup theory, and several related studies have been conducted (e.g., \cite{groupComplexity,localDivisor,q-theory}).
  Because the wreath product of two semigroups $N$ and $G$ is defined as the transformation semigroup $(N\times G,N^G\rtimes G)$,
  our decomposition of the form $N^G\rtimes G$ with a group $G$ provides a partial Krohn-Rhodes decomposition of syntactic monoids.
  Note that DFAs can be regarded as transformation semigroups, and therefore, the Krohn-Rhodes decomposition has been studied in the context of formal language theory.
  In particular, the holonomy decomposition is known as a decomposition of finite state automata \cite{Eilenberg,Cascade}.
  However, the holonomy decomposition is a decomposition of automata, not a decomposition of syntactic monoids.
\end{enumerate}

The remaining sections of this paper are organized as follows.
Section~2 provides basic definitions of monoids, languages, and periods. In addition, we provide some examples of periodic regular languages that will be discussed in the later sections.
In Section~3, we first prove the main theorem (Theorem~\ref{th-main}). In the latter part of the section, we focus on periods with respect to a given alphabet $\Sigma$ and explain that the syntactic monoid can be decomposed into monoids corresponding to each residual of the period (Theorem~\ref{th-residualMonoid}).
In Section~4, we discuss the applications (i) and (ii) described above in detail. Specifically, we discuss the connection between the periods of languages we defined and the periods in the context of Markov chains (Theorem~\ref{th-maxperiod}). Moreover, we extend the characterization of zero-one languages presented in \cite{zero-one} by applying our results (Theorem~\ref{th-ex-zeroone}).
In section~5, we consider the wreath products described in the application (iii), and provide a partial Krohn-Rhodes decomposition of periodic regular languages (Theorem~\ref{th-main2}).

\section{Preliminaries}
Let $|X|$ denote the cardinality of a set $X$. For sets $X$ and $Y$, let $X\sqcup Y$ and $Y^X$ denote the disjoint union of $X$ and $Y$, and the set of all functions from $X$ to $Y$, respectively. A function from $X$ to itself is called a {\it transformation} on $X$. For a positive integer $K$, $K$ is sometimes regarded as the set $\{0,\ldots,K-1\}$.

\subsection{Monoids}
A {\it monoid} is a set $M$ equipped with an associative binary operation $\cdot:M\times M\to M$, and containing the identity element $e_M\in M$. If the monoid $M$ is clear from the context, $e_M$ is simply denoted as $e$.
For monoids $M$ and $N$, we say that $h:M\to N$ is a monoid homomorphism if $h$ satisfies: (i) $h(e_M)=e_N$, and (ii) $h(m_1\cdot m_2)=h(m_1)\cdot h(m_2)$ for each $m_1,m_2\in M$.
We say that a subset $N'\subseteq N$ is a submonoid of $N$ if $N'$ also forms a monoid.
Note that for every monoid homomorphism $h:M\to N$, the homomorphic image $h(M)$ is a submonoid of $N$.
Moreover, $M$ is isomorphic to the submonoid $h(M)$ if $h$ is injective.
Therefore, we also say that $M$ is a submonoid of $N$ if there exists an injective monoid homomorphism $h:M\to N$.
In this case, we identify $M$ with $h(M)$, and each $m\in M$ with $h(m)\in N$ if the embedding $h:M\to N$ is clear from the context.
We say that $N$ is a {\it quotient} of $M$ if there exists a surjective monoid homomorphism $\psi:M\to N$.
Also, $N$ is a {\it divisor} of $M$ if $N$ is a quotient of a submonoid of $M$.
\begin{example}\label{ex-basic_monoids}
  Let $K$ be a positive integer.
  The followings are examples of monoids used in this paper. \begin{itemize}
  \item Let $U_K=\{e,\iota_1,\ldots,\iota_K\}$ be the monoid such that all of non-identity elements are right-zero. That is, $s\cdot \iota_i=\iota_i$ for each $1\le i\le K$ and $s\in U_K$.
  \item Let $\overline{U}_K=\{e,\iota_1,\ldots,\iota_K\}$ be the monoid such that all of non-identity elements are left-zero.
  \item Let $C_K=\{0,\ldots,K-1\}$ be the cyclic group of order $K$. That is, $i\cdot j=i+j\!\!\mod K$ for each $i,j\in C_K$. We write $+$ for the operation of $C_K$.
  \item Let $\mathcal{S}_K$ be the symmetric group of order $K!$. That is, $\mathcal{S}_K$ consists of all permutations on $K$, and the operation $\cdot$ is the composition of functions.
  \item Let $\mathcal{T}_K$ be the monoid such that the carrier set consists of all transformations on $K$, and the operation is the composition of functions.\bqed
\end{itemize}
\end{example}
Note that any finite monoid is a submonoid of $\TK$ for some integer $K$. This is because every element $m$ in a monoid $M$ can be regarded as the transformation $\tau:M\to M$ such that $\tau(s)=s\cdot m$ for each $s\in M$.

Let $M$ be a finite monoid and $S\subseteq M$ be a generator of $M$. We say that $(V,E)$ is the {\it Cayley graph} of $M$ where $V=M$ is the set of vertices and $E=\{(m_1,s,m_2)\in M\times S\times M\mid m_1\cdot s=m_2  \}$ is the set of edges labeled by $S$.

\subsection{Semidirect Products}
Let $X$ be a set and $M$ be a monoid. A {\it left action} on $X$ from $M$ is a function $*:M\times X\to X$ satisfying: (i) $e_M* x=x$, and (ii) $m_1*(m_2*x)=(m_1\cdot m_2)*x$ for each $m_1,m_2\in M$ and $x\in X$. A {\it right action} is defined as the dual of a left action.
In particular, the operation of $M$ is a left (or right) action on $M$ from $M$. When considering an action on $M$ from $M$, the action refers to the operation of $M$ unless otherwise specified.

For monoids $M$ and $N$, a left action $*$ on $M$ from $N$ is {\it distributive} if $n*(m_1\cdot m_2)=(n*m_1)\cdot(n*m_2)$ for all $n\in N$ and $m_1,m_2\in M$.
In this paper, all actions on monoids are supposed to be distributive.
A left action $*$ on $M$ from $N$ is {\it unitary} if $n* e_M=e_M$ for all $n\in N$.

Let $\otimes$ and $\oplus$ be the operations of monoids $M$ and $N$, respectively. For a distributive left action $*:N\times M\to M$, the monoid with operation $\cdot$ on $M\times N$ defined as\begin{equation}
(m_1,n_1)\cdot (m_2,n_2) = (m_1\otimes(n_1*m_2),n_1\oplus n_2)\label{eq:semidirect*1}
\end{equation}
for each $(m_1,n_1),(m_2,n_2)\in M\times N$
is called the {\it semidirect product} (with respect to $*$) of $M$ and $N$, and denoted by $M\rtimes_* N$. By the distributivity of the action $*$, the operation $\cdot$ of $M\rtimes_* N$ is associative, and
\begin{equation}\label{eq:closedform}
(m_1, n_1) \cdot (m_2, n_2) \cdot \cdots \cdot (m_k, n_k)
   = (  \bigotimes_{1 \le i \le k} (\bigoplus_{1 \le j < i} {n_j}) * m_i ,  \bigoplus_{1 \le i \le k} {n_i} )
\end{equation}
for each $(m_1, n_1) ,\ldots ,(m_k, n_k)\in M\times N$.

\begin{example}\label{ex-semidirect}
  \begin{itemize}
    \item The direct product $M\times N$ of any two monoids $M$ and $N$ is the semidirect product with respect to the {\it trivial action} $*$ where $*$ is defined as $n* m=m$ for each $n\in N,m\in M$.
    See the closed form \eqref{eq:closedform}, whose right-hand becomes $(\bigotimes_{1 \le i \le k}  m_i ,  \bigoplus_{1 \le i \le k} {n_i})$ when $n*m=m$.
    That is, each element $m\in M$ is not affected by any $n\in N$ in this action~$*$.
    \item The symmetric group $\mathcal{S}_3$ is isomorphic to $C_3\rtimes_{*}C_2$ where $*$ is the unitary action defined as $0* m=m$ and $1*m=-m$ for each $m\in C_3$.\bqed
  \end{itemize}
\end{example}

Let $M$ be a monoid with an operation $\otimes$, and $*:Y\times N\to Y$ be a right action on a set $Y$ from a monoid $N$. Let $M^Y$ be the monoid defined as $(f\cdot g)(y)=f(y)\otimes g(y)$ for each $f,g\in M^Y$.
Then, the left action $\circledast:N\times M^Y\to M^Y$ is induced as \begin{equation}
(n\circledast f)(y)=f(y*n) \label{eq:act-function*2}
\end{equation}
for each $n\in N,y\in Y$ and $f\in M^Y$.
That is, $\circledast$ is a pointwise extension of $*$ from $M$ to $M^Y$.
We let $M^Y\rtimes N$ denote the semidirect product of $M^Y$ and $N$ with respect to this action $\circledast$.
Note that when considering the case $Y=N$ (i.e., considering $M^N\rtimes N$), the action $*:Y\times N\to Y$ refers to the operation on $N$.
This special semidirect product is remarkable because of the following known fact.

\begin{proposition}\label{prop-semidirect}
  Let $M$ and $N$ be monoids.
  The semidirect product $M\rtimes_* N$ is a submonoid of $M^N\rtimes N$ for every unitary left action $*:N\times M\to M$. In particular, $M\times N$ is a submonoid of $M^N\rtimes N$.
\end{proposition}
\begin{proof}
  We define $h:M\rtimes_* N\to M^N\rtimes N$ as $h(m,n)=(f,n)$ for each $m\in M$ and $n\in N$ where $f\in M^N$ is defined as $f(n')=n'*m$ for each $n'\in N$. Since $f(e_N)=e_N*m=m$ for each $m\in M$, $h$ is injective. Because $*$ is unitary, $h(e_M,e_N)=(f_e,e_N)=e_{M^N\rtimes N}$ where $f_e$ is the constant function that maps every $n'\in N$ to $e_M$. Furthermore,
  \begin{eqnarray*}
    h((m_1,n_1))\cdot h((m_2,n_2))&=&(f_1,n_1)\cdot (f_2,n_2)\\
    &=&(f_1\otimes (n_1\oplus f_2),n_1\oplus n_2)
  \end{eqnarray*}
  where $f_1$ (resp. $f_2$) is defined as $f_1(n')=n'*m_1$ (resp. $f_2(n')=n'*m_2$) for each $n'\in N$. Because $(f_1\otimes (n_1\oplus f_2))(n')=f_1(n')\otimes (n_1\oplus f_2)(n')=(n'*m_1)\otimes (n'*(n_1*m_2))=n'*(m_1\otimes (n_1*m_2))$ for each $n'\in N$,
  \begin{eqnarray*}
    h((m_1,n_1))\cdot h((m_2,n_2))&=&
  (f_1\otimes (n_1\oplus f_2),n_1\oplus n_2)\\
  &=&h(m_1\otimes (n_1*m_2),n_1\oplus n_2)\\
  &=&h((m_1,n_1)\cdot (m_2,n_2))
  \end{eqnarray*}
  holds. Therefore, $h$ is an injective monoid homomorphism and $M\rtimes_* N$ is a submonoid of $M^N\rtimes N$.
  It also holds for the direct product $M\times N$ because the trivial action $N\times M\to M$ is a unitary action.\qed
\end{proof}

\subsection{Periods of Regular Languages}
Let $\Sigma$ be a (finite) alphabet.
For $w\in\Sigma^*$, $|w|$ denotes the length of $w$. We let $|w|_a$ denote the number of occurrences of $a\in \Sigma$ in $w\in\Sigma^*$, and $|w|_\Gamma=\sum_{a\in \Gamma}|w|_a$ for $\Gamma\subseteq \Sigma$. For example, $|w|_a=3, |w|_{\{a,b\}}=4$ and $|w|=|w|_{\Sigma}=5$ with $w=aabac\in\{a,b,c\}^*$.

For a regular language $L\subseteq \Sigma^*$, we say that a monoid $M$ {\it fully recognizes} $L$ if there is a surjective monoid homomorphism $\eta:\Sigma^*\to M$ such that $L=\eta^{-1}(S)$ with some $S\subseteq M$.
The smallest monoid that fully recognizes $L$ is called the {\it syntactic monoid} of $L$, and the corresponding homomorphism is called the {\it syntactic morphism} of $L$.
The uniqueness of the syntactic monoid and the syntactic morphism is guaranteed by the following proposition (see \cite{Lallement} in detail).
\begin{proposition}\label{prop-recognize}
Let $L\subseteq\Sigma^*$ be a regular language, and let $M_L$ and $\eta_L$ be the syntactic monoid and the syntactic morphism of $L$, respectively. For any monoid $M$, if $M$ fully recognizes $L$ with a homomorphism $\eta$, then there exists a surjective homomorphism $\psi:M\to M_L$ such that $\eta_L=\psi\circ \eta$.
That is, the following commutative diagram holds:
\begin{equation*}
  \xymatrix{
  {\Sigma^*}\ar@{->>}[r]^-{\forall\eta}\ar[rd]_-{\eta_{L}}\ar@{}<2.0ex>\ar@{->>}[rd]|{}&{M}\ar@{..>>}[d]^-{\exists\psi}\\
  &{M_{L}}~.
  }\vspace{-18pt}
\end{equation*}
\bqed
\end{proposition}
Because $\Sigma^*$ is generated by $\Sigma$, the syntactic monoid $M_L$ is generated by $\eta_L(\Sigma)$. When we illustrate the Cayley graph of $M_L$, the label $\eta_L(a)$ is often abbreviated as $a$ for each $a\in\Sigma$.
{
It is well known that the syntactic monoid of $L$ is isomorphic to the {\it transition monoid}\footnote{{For a DFA $\mathcal{A}$ with the set of states $Q$, the transition monoid of $\mathcal{A}$ is defined as $T(\mathcal{A})=\{\tau_w\in Q\to Q\mid w\in \Sigma^*\}$ where $\tau_w(q)=q'$ iff $q\xrightarrow{w}q'$ is the transition of $\mathcal{A}$ for each $q,q'\in Q$.}} of the minimal DFA of $L$.}

We define periods of regular languages, which is the key concept of this paper. As discussed in Section 1, the concept of periods is inspired by studies on finite state Markov chains. Therefore, it is natural to define periods based on the graph structure of a DFA.
However, there are more than one DFAs that recognize the same regular language.
In this paper, we opt for the Cayley graph of the syntactic monoid to define periods. (Note that the Cayley graph of the syntactic monoid of a regular language $L$ can be regarded as a DFA recognizing $L$.)
This is because the syntactic monoid more appropriately represents the periodicity of a regular language than the minimal DFA and other DFAs.
{For example, Lemmas \ref{lem-disjoint} and \ref{lem-syntPeriod}, mentioned later, demonstrate essential properties of periods, but the same results do not hold for the minimal DFA.}
We will discuss this aspect in detail in Section 4.1.

\begin{definition}
  Let $L\subseteq \Sigma^*$ be a regular language, and let $M_L$ and $\eta_L$ be the syntactic monoid and the syntactic morphism of $L$, respectively.
  For a non-empty subset $\Gamma\subseteq \Sigma$, we say that $L$ has a {\rm period} $P$ with respect to $\Gamma$ if $P$ satisfies: for every $w\in\Sigma^*$ such that $t\cdot \eta_L(w)=t$ for some $t\in M_L$, $|w|_{\Gamma}$ is a multiple of $P$.\bqed
\end{definition}

{By the definition of period, every language has period one with respect to each subset of $\Sigma$. We are mainly interested in the maximum number of all the periods for each subset of $\Sigma$. For example, all periods mentioned in Example~\ref{ex-language} below are maximum periods.
We often say that a language $L\subseteq\Sigma^*$ is {\it periodic} if there is a period greater than one with respect to some non-empty subset $\Gamma\subseteq \Sigma$.}
In the rest of the paper, $M_L$ and $\eta_L$ always represent the syntactic monoid and the syntactic morphism of a regular language $L$.

\begin{example}\label{ex-language}
  The followings are examples of periodic regular languages.
\begin{enumerate}[(1)]
\item Let $L_1\subseteq \{a,b\}^*$ be the language defined by DFA $\mathcal{A}_1$ (see Figure 1-1). The syntactic monoid of $L_1$ is $M_{L_1}=C_2\times C_2$, and the Cayley graph of $M_{L_1}$ has the same shape as $\mathcal{A}_1$. Therefore, $L_1$ has period $2$ with respect to $\{a\},\{b\}$ and $\{a,b\}$.
\item  Let $L_2\subseteq \{a,b\}^*$ be the language defined by DFA $\mathcal{A}_2$ (see Figure 1-2). The syntactic monoid of $L_2$ is $M_{L_2}=\mathcal{S}_3$, and the Cayley graph of $M_{L_2}$ is shown in Figure 2-2.
We can easily show that $\eta_{L_2}(a)$ is an odd permutation and $\eta_{L_2}(b)$ is an even permutation. Therefore, $L_2$ has period $2$ with respect to $\{a\}$.
\item Let $L_3\subseteq \{a,b\}^*$ be the language defined by DFA $\mathcal{A}_3$ (see Figure 1-3). By the shape of the Cayley graph of $M_{L_3}$ shown in Figure 2-3, $L_3$ has period $2$ with respect to $\{a,b\}$.
\end{enumerate}
\end{example}
\begin{minipage}[b]{0.3\linewidth}
  \centering
  \begin{tikzpicture}[auto]
    \node [scale=0.65] (s0) [state, initial, initial above, initial text=,accepting] at (0,0) {};
    \node [scale=0.65] (s1) [state] at (1.4,0) {};
    \node [scale=0.65] (s2) [state] at (0,-1.4) {};
    \node [scale=0.65] (s3) [state] at (1.4,-1.4) {};

    \path [-stealth, thick]
      (s0) edge [bend left] node {$a$} (s1)
      (s1) edge node {$a$} (s0)
      (s2) edge node {$a$} (s3)
      (s3) edge [bend left] node {$a$} (s2)
      (s0) edge node {$b$} (s2)
      (s2) edge [bend left] node {$b$} (s0)
      (s1) edge [bend left] node {$b$} (s3)
      (s3) edge node {$b$} (s1);
  \end{tikzpicture}
  \\\vspace{2pt}{\bf Figure 1-1.} DFA $\mathcal{A}_1$
\end{minipage}
\begin{minipage}[b]{0.3\linewidth}
  \centering
  \begin{tikzpicture}[auto]
    \node [scale=0.65] (s0) [state, initial, initial above, initial text=,accepting] at (0,0) {};
    \node [scale=0.65] (s2) [state] at (0,-1.4) {};
    \node [scale=0.65] (s3) [state] at (1.4,-1.4){};

    \path [-stealth, thick]
      (s0) edge [loop right] node {$a$} ()
      (s2) edge node {$b$} (s0)
      (s0) edge node {$b$} (s3)
      (s3) edge [bend left] node {$a,b$} (s2)
      (s2) edge node {$a$} (s3);
  \end{tikzpicture}
  \\\vspace{2pt}{\bf Figure 1-2.} DFA $\mathcal{A}_2$
\end{minipage}
\begin{minipage}[b]{0.35\linewidth}
  \begin{tikzpicture}[auto]
    \node [scale=0.65] (s0) [state] at (0,0) {};
    \node [scale=0.65] (s1) [state, initial, initial above, initial text=] at (1.4,0) {};
    \node [scale=0.65] (s2) [state, accepting] at (0,-1.4) {};
    \node [scale=0.65] (s3) [state, accepting] at (2.8,-1.4) {};

    \path [-stealth, thick]
      (s3) edge [loop above] node {$a,b$} ()
      (s1) edge [bend right] node [below] {$b$} (s3)
      (s1) edge [bend left] node [below] {$a$} (s2)
      (s0) edge node {$a,b$} (s2)
      (s2) edge [bend left] node {$a,b$} (s0);
  \end{tikzpicture}
  \centering
  \\\vspace{16pt}\hspace{18pt}{\bf Figure 1-3.} DFA $\mathcal{A}_3$
\end{minipage}

\vspace{5pt}
\begin{minipage}[b]{0.45\linewidth}
  \begin{tikzpicture}[auto]
    \node [scale=0.65] (s0) [state] at (0,0) {};
    \node [scale=0.65] (s1) [state] at (1.4,0) {};
    \node [scale=0.65] (s2) [state] at (2.8,0) {};
    \node [scale=0.65] (s3) [state] at (0,-1.4) {};
    \node [scale=0.65] (s4) [state] at (1.4,-1.4) {};
    \node [scale=0.65] (s5) [state] at (2.8,-1.4) {};

    \path [-stealth, thick]
      (s0) edge node [below] {$b$} (s1)
      (s1) edge node [below] {$b$} (s2)
      (s2) edge [bend right] node [above] {$b$} (s0)
      (s4) edge node [above] {$b$} (s3)
      (s5) edge node [above] {$b$} (s4)
      (s3) edge [bend right] node [below] {$b$} (s5)
      (s0) edge [bend left] node {$a$} (s3)
      (s3) edge node {$a$} (s0)
      (s1) edge [bend left] node {$a$} (s4)
      (s4) edge node {$a$} (s1)
      (s2) edge [bend left] node {$a$} (s5)
      (s5) edge node {$a$} (s2);
  \end{tikzpicture}
  \centering
  \\\vspace{3pt}{\bf Figure 2-2.} Cayley graph of $M_{L_2}$
\end{minipage}
\begin{minipage}[b]{0.45\linewidth}
  \begin{tikzpicture}[auto]
    \node [scale=0.65] (s0) [state] at (0,0) {};
    \node [scale=0.65] (s1) [state] at (1.4,0) {};
    \node [scale=0.65] (s2) [state] at (0,-1.4) {};
    \node [scale=0.65] (s3) [state] at (2.8,-1.4) {};
    \node [scale=0.65] (s4) [state] at (2.8,0) {};

    \path [-stealth, thick]
      (s4) edge [bend left] node {$a,b$} (s3)
      (s3) edge node {$a,b$} (s4)
      (s1) edge [bend right] node [below] {$b$} (s3)
      (s1) edge [bend left] node [below] {$a$} (s2)
      (s0) edge node {$a,b$} (s2)
      (s2) edge [bend left] node {$a,b$} (s0);
  \end{tikzpicture}
  \centering
  \\\vspace{18pt}\hspace{6pt}{\bf Figure 2-3.} Cayley graph of $M_{L_3}$
\end{minipage}

\section{Decompositions of Periodic Regular Languages}
Let $L\subseteq \Sigma^*$ be a regular language with periods $P_1,\ldots,P_n> 1$ with respect to $\Gamma_1,\ldots,\Gamma_n\subseteq \Sigma$, respectively. We define the monoid homomorphism $\rho:\Sigma^*\to (C_{P_1}\times \cdots\times C_{P_n})$ as $\rho(w)=(r_1,\ldots,r_n)$ where \[r_i=|w|_{\Gamma_i}\!\!\mod P_i\] for each $1\le i\le n$.
We call $\rho(w)$ the {\it residual} of $w$ (modulo $P_1,\ldots, P_n$).
For example, let $\Sigma=\{a,b\}$ and let $L$ have periods $P_1=2$ and $P_2=2$ with respect to $\Gamma_1=\{a\}$ and $\Gamma_2=\{a,b\}$, respectively.
Then, $\rho(aabac)=(1,0)$ holds.

The next lemma states that the syntactic monoid is partitioned into $N_{\vb{r}}=\eta_L(\{w\in\Sigma^*\mid \rho(w)=\vb{r}\})$ for residuals $\vb{r}\in C_{P_1}\times \cdots\times C_{P_n}$.
{
Note that this property does not hold for the minimal DFA. (Compare the DFAs in Figures 1--3 and 2--3 in Example~\ref{ex-language}.)}

\begin{lemma}\label{lem-disjoint}
  Let $L\subseteq \Sigma^*$ be a regular language, and $M_L$ and $\eta_L$ be the syntactic monoid and the syntactic morphism of $L$.
  Let $L$ have periods $P_1,\ldots,P_n> 1$ with respect to $\Gamma_1,\ldots,\Gamma_n\subseteq \Sigma$, respectively. Then, it holds that\[
  M_L=\bigsqcup_{\vb{r}\in C_{P_1}\times \cdots\times C_{P_n}}N_{\vb{r}}
  \]
  where each $N_{\vb{r}}\subseteq M_L$ is defined as $N_{\vb{r}}=\eta_L(\{w\in\Sigma^*\mid \rho(w)=\vb{r}\})$.
\end{lemma}
\begin{proof} It holds that $M_L=\bigcup_{\vb{r}\in C_{P_1}\times \cdots\times C_{P_n}}N_{\vb{r}}$ because $\eta_L$ is surjective. Therefore, we show that $N_{\vb{r}_1}\cap N_{\vb{r}_2}=\emptyset$ for each $\vb{r}_1\ne \vb{r}_2$.
We assume that $N_{\vb{r}_1}\cap N_{\vb{r}_2}\ne \emptyset$ for some $\vb{r}_1\ne \vb{r}_2$. Then, there exist $w_1,w_2\in\Sigma^*$ and $\Gamma_i\in\{\Gamma_1,\ldots,\Gamma_n\}$ such that $\eta_L(w_1)=\eta_L(w_2)$ and $|w_1|_{\Gamma_i}\ne |w_2|_{\Gamma_i}~{\rm mod}~ P_i$.
By concatenating an appropriate suffix to $w_1$ and $w_2$, we can assume that $|w_1|_{\Gamma_i}$ is a multiple of $P_i$, and $|w_2|_{\Gamma_i}$ is not a multiple of $P_i$ without loss of generality.

Note that there are $j_1\ge 0$ and $j_2\ge 1$ such that $\eta_L(w_1)^{j_1}\cdot \eta_L(w_1)^{j_2}=\eta_L(w_1)^{j_1}$ because $M_L$ is finite. Since $\eta_L(w_1)=\eta_L(w_2)$, it holds that\begin{eqnarray*}
\eta_L(w_1)^{j_1}&=& \eta_L(w_1)^{j_1}\cdot \eta_L(w_1)^{j_2}\\
&=&\eta_L(w_1)^{j_1}\cdot \eta_L(w_1)^{j_2-1}\cdot \eta_L(w_2)\\
&=&\eta_L(w_1)^{j_1}\cdot \eta_L(w_1^{j_2-1}w_2)~.
\end{eqnarray*}
However, $|w_1^{j_2-1}w_2|_{\Gamma_i}$ is not a multiple of $P_i$, and it contradicts the fact that $L$ has period $P_i$ with respect to $\Gamma_i$.\qed
\end{proof}

Let $N_{\vb{r}}$ be the subset defined in Lemma \ref{lem-disjoint} for each $\vb{r}\in C_{P_1}\times \cdots\times C_{P_n}$.
We define $\overline{\rho}:M_L\to C_{P_1}\times \cdots\times C_{P_n}$ as $\overline{\rho}(t)=\rho(w)$
with any $w\in \Sigma^*$ such that $t = \eta_L(w)$.
By Lemma~\ref{lem-disjoint}, $\overline{\rho}$ is well-defined and is a monoid homomorphism. We call $\overline{\rho}(t)$ the {\it residual} of $t$ (modulo $P_1,\ldots, P_n$).

\subsection{Semidirect Product Decompositions with Cyclic Groups}
Our goal is to show that $M_L$ can be decomposed into a submonoid of
$N^{C_{P_1} \times \cdots\times C_{P_n}} \rtimes (C_{P_1} \times \cdots \times C_{P_n})$ for some appropriate monoid $N$ described below.
This decomposition means that when a residual $\vb{r}$ (modulo $P_1,\ldots,P_n$) is given,
we can obtain the fragment of $M_L$ with respect to $\vb{r}$ ($N_{\bf r}$ in Figure 3) as follows:
First, retrieve information from $N$ (by using $\vb{r}$ as a retrieval key)
and then take the semidirect product of the retrieved fragment and $\vb{r}$.
By Lemma~\ref{lem-disjoint}, $M_L$ can be represented as the disjoint union of $N_{\vb{r}}$ for every residual $\vb{r}$ modulo $P_1,\ldots, P_n$.
That is, $N_{\vb{r}}$ is exactly the information on $M_L$ with respect to $\vb{r}$ mentioned above.
{
To embed $N_{\vb{r}}$ for every residual $\vb{r}$ into $N$, we set $N=\TK$ for an appropriate number $K\ge 1$. In fact, taking $K=\max_{\vb{r}}\{|N_{\vb{r}}|\}$ suffices.}

\begin{figure}[h]
  \centering
\includegraphics[height=110pt]{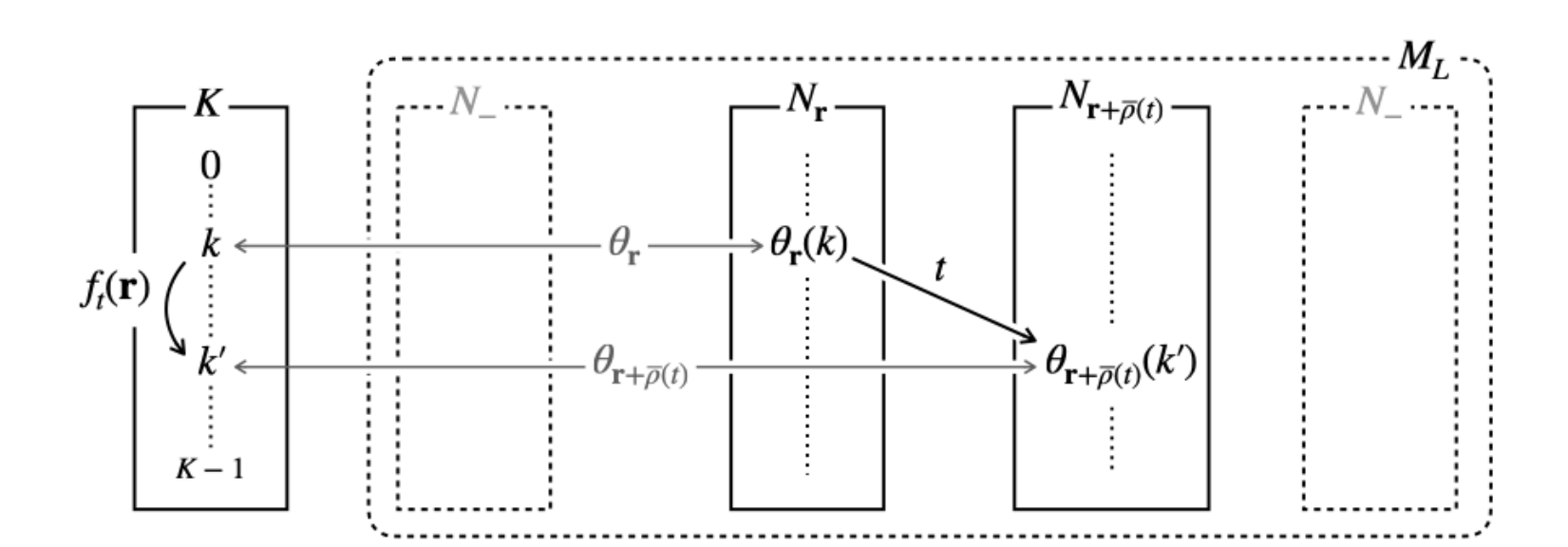}
\caption*{{\bf Figure 3.} The definition of $f_t$}
\end{figure}
\begin{theorem}\label{th-main}
  Let $L\subseteq \Sigma^*$ be a regular language, and $M_L$ and $\eta_L$ be the syntactic monoid and the syntactic morphism of $L$. For any non-empty subsets $\Gamma_1,\ldots,\Gamma_n\subseteq \Sigma$, the following conditions are equivalent:\begin{enumerate}[(i)]
  \item $L$ has periods $P_1,\ldots,P_n$ with respect to $\Gamma_1,\ldots,\Gamma_n\subseteq \Sigma$, respectively.
  \item $M_L$ is a submonoid of $\TK^{C_{P_1}\times\cdots\times C_{P_n}}\rtimes (C_{P_1}\times\cdots\times C_{P_n})$ where $K=\max\{|N_{\vb{r}}|\mid \vb{r}\in C_{P_1}\times\cdots\times C_{P_n}\}$. Furthermore, $\eta_L(w)\in \TK^{C_{P_1}\times\cdots\times C_{P_n}}\times \{\rho(w)\}$ for each $w\in\Sigma^*$.
  \end{enumerate}
\end{theorem}
\begin{proof}
Suppose (ii) holds.
Note that $(f_1,\vb{r}_1)\cdot (f_2,\vb{r}_2)=(f_1\cdot (\vb{r_1}+f_2),\vb{r}_1+\vb{r}_2)$ for each $(f_1,\vb{r}_1),(f_2,\vb{r}_2)\in M_L$.
(See \eqref{eq:semidirect*1} in the definition of $\rtimes_*$ and \eqref{eq:act-function*2} in the definition of $\rtimes$.)
Therefore, if $t\cdot \eta_L(w)=t$ for some $t\in M_L$ and $w\in \Sigma^*$, then $\rho(w)=(0,\ldots,0)$ because $(f_1,\vb{r}_1)\cdot (f_2,\vb{r}_2)=(f_1,\vb{r}_1)$ implies $\vb{r}_2=(0,\ldots,0)$.
Hence, $|w|_{\Gamma_i}=0\!\mod P_i$ for each $1\le i\le n$, and (i) holds.

We show (i)$\Rightarrow$(ii).
It suffices to prove that if $L$ has periods $P_1,\ldots,P_n$ with respect to $\Gamma_1,\ldots,\Gamma_n\subseteq \Sigma$, then there is an injective homomorphism $\Can:M_L\to \TK^{C_{P_1}\times\cdots\times C_{P_n}}\rtimes (C_{P_1}\times\cdots\times C_{P_n})$ such that $\Can(t)\in  \TK^{C_{P_1}\times\cdots\times C_{P_n}}\times \{{\overline{\rho}}(t)\}$ for each $t \in M_L$.

Let $\theta_{\vb{r}}:|N_{\vb{r}}|\to N_{\vb{r}}$ be an arbitrary bijection for each $\vb{r}\in C_{P_1}\times\cdots\times C_{P_n}$.
Intuitively, $\theta^{-1}_{\vb{r}}$ provides a total ordering of $N_{\vb{r}}$.
For each $t\in M_L$, define $\Can(t)=(f_t,\overline{\rho}(t))\in \TK^{C_{P_1}\times\cdots\times C_{P_n}}\rtimes (C_{P_1}\times\cdots\times C_{P_n})$ where
\[f_t(\vb{r})(k)=\begin{cases}
{\theta_{{\vb{r}}+\overline{\rho}(t)}}^{-1}(\theta_{\vb{r}}(k)\cdot t)&\mbox{if~~~}k<|N_{\vb{r}}|~,\\
k&\mbox{if~~~}k\ge |N_{\vb{r}}|
\end{cases}
\]for each ${\vb{c}}\in C_{P_1}\times\cdots\times C_{P_n}$ and $0\le k<K$.
We show that $\Can$ is a homomorphism, that is, $\Can(s)\cdot \Can(s')=\Can(t)$ for each $s,s',t\in M_L$ such that $t=s\cdot s'$.
Let $\Can(s)=(f_s,\overline{\rho}(s))$ and $\Can(s')=(f_{s'},\overline{\rho}(s'))$. Then,
\begin{eqnarray*}
  \Can(s)\cdot \Can(s')
  &=&(f_s,\overline{\rho}(s))\cdot (f_{s'},\overline{\rho}(s'))\\
  &=& (f_s\cdot ( \overline{\rho}(s)\circledast f_{s'}),\overline{\rho}(s)+ \overline{\rho}(s'))\\
  &=& (f_s\cdot ( \overline{\rho}(s)\circledast f_{s'}),\overline{\rho}(t))
\end{eqnarray*}
where $(f_s\cdot ( \overline{\rho}(s)\circledast f_{s'}))(\vb{r})$ with $\vb{r}\in C_{P_1}\times\cdots\times C_{P_n}$ is a transformation such that
\begin{eqnarray*}
  (f_s\cdot ( \overline{\rho}(s)\circledast f_{s'}))(\vb{r})(k)
  &=&(( \overline{\rho}(s)\circledast f_{s'})(\vb{r})\circ f_s(\vb{r}))(k)\\
  &=&f_{s'}(\vb{r}+\overline{\rho}(s))(f_s(\vb{r})(k))\\
  &=&f_{s'}(\vb{r}+\overline{\rho}(s))({\theta_{\vb{r}+\overline{\rho}(s)}}^{-1}(\theta_{\vb{r}}(k)\cdot s))\\
  &=&{\theta_{\vb{r}+\overline{\rho}(s)+\overline{\rho}(s')}}^{-1}
  ({\theta_{\vb{r}+\overline{\rho}(s)}}({\theta_{\vb{r}+\overline{\rho}(s)}}^{-1}(\theta_{\vb{r}}(k)\cdot s))\cdot s')\\
  &=&{\theta_{\vb{r}+\overline{\rho}(s\cdot s')}}^{-1}(\theta_{\vb{r}}(k)\cdot (s\cdot s'))\\
  &=&{\theta_{\vb{r}+\overline{\rho}(t)}}^{-1}(\theta_{\vb{r}}(k)\cdot t)
\end{eqnarray*}
for each $0\le k<|N_{\vb{r}}|$. Therefore, $\Can(s)\cdot \Can(s')=\Can(t)$.

Finally, we show that $\Can$ is injective. Let $\Can(t)=(f_t,\vb{r})$ and $\Can(t')=(f_{t'},\vb{r}')$. By the definition of $f_t$, it holds that
\begin{eqnarray*}
f_t(\vb{0})({\theta_{\vb{0}}}^{-1}(e_{M_L}))={\theta_{\vb{r}}}^{-1}(\theta_{\vb{0}}({\theta_{\vb{0}}}^{-1}(e_{M_L}))\cdot t)={\theta_{\vb{r}}}^{-1}(e_{M_L}\cdot t)={\theta_{\vb{r}}}^{-1}(t)
\end{eqnarray*}
where $\vb{0}=(0,\ldots,0)$.
In the same way, $f_{t'}(\vb{0})({\theta_{\vb{0}}}^{-1}(e_{M_L}))={\theta_{\vb{r}'}}^{-1}(t')$.
Now, if $\Can(t)=\Can(t')$, then $f_t=f_{t'}$ and $\vb{r}=\vb{r}'$.
Therefore, \[{\theta_{\vb{r}}}^{-1}(t)=f_t(\vb{0})({\theta_{\vb{0}}}^{-1}(e_{M_L}))=f_{t'}(\vb{0})({\theta_{\vb{0}}}^{-1}(e_{M_L}))={\theta_{\vb{r}'}}^{-1}(t')={\theta_{\vb{r}}}^{-1}(t').\]
Thus, $t=\theta_{\vb{r}}({\theta_{\vb{r}}}^{-1}(t))=\theta_{\vb{r}}({\theta_{\vb{r}}}^{-1}(t'))=t'$ holds.\qed
\end{proof}

{This theorem claims that the syntactic monoid can be decomposed into the first component $\TK^{C_{P_1}\times\cdots\times C_{P_n}}$ and the second component $C_{P_1}\times\cdots\times C_{P_n}$ of the semidirect product stated in (ii). The significance of this decomposition is as follows.
For the second component $C_{P_1}\times\cdots\times C_{P_n}$, the last condition $\eta_L(w)\in \TK^{C_{P_1}\times\cdots\times C_{P_n}}\times \{\rho(w)\}$ in the statement of the theorem is crucial.
Intuitively, this condition says that the periodicity of the syntactic monoid is explicitly extracted as $C_{P_1}\times\cdots\times C_{P_n}$.
Next, let us consider the first component $\TK^{C_{P_1}\times\cdots\times C_{P_n}}$.
If we naively consider the behavior of each element in $M_L$ as an action on the set $M_L$, the corresponding transformation monoid is $\mathcal{T}_{|M_L|}$.
On the other hand, Theorem~\ref{th-main} states that by fixing a residual ${\bf r}$, each element can be described as an action on a set of size at most $K<|M_L|$, and the order of $\TK$ is
much smaller than that of $\mathcal{T}_{|M_L|}$.
More intuitively, to describe a periodic regular language, it suffices to have at most $K$ states for each residue ${\bf r}$.
Nevertheless, $\TK$ is still a large monoid.
The possibility of replacing the first component $\TK^{C_{P_1}\times\cdots\times C_{P_n}}$ with a simpler monoid needs further investigation.}

In this paper, the injective homomorphism $\Can:M_L\to \TK^{C_{P_1}\times\cdots\times C_{P_n}}$ in the proof of Theorem~\ref{th-main} is called the {\it canonical homomorphism}.

\begin{example}
As mentioned in Example \ref{ex-language}-(1), $L_1$ has period $2$ with respect to both of $\Gamma_1=\{a\}$ and $\Gamma_2=\{b\}$. Therefore, $M_{L_1}$ is a submonoid of $\mathcal{T}_1^{C_2\times C_2}\rtimes (C_2\times C_2)$.
(Where $K=1$ for $\mathcal{T}_K$ because $\max\{|N_{(r_1,r_2)}|\mid r_1,r_2\in C_2\}=\max\{1\}=1$.)
In fact, $M_{L_1}=C_2\times C_2$ is isomorphic to $\mathcal{T}_1^{C_2\times C_2}\rtimes (C_2\times C_2)$.\bqed
\end{example}
\begin{example}
As mentioned in Example \ref{ex-language}-(2), $L_2$ has period $2$ with respect to $\Gamma_1=\{a\}$. Therefore, $M_{L_2}$ is a submonoid of $\mathcal{T}_3^{C_2}\rtimes C_2$.
Note that $M_{L_2}=\mathcal{S}_3=C_3\rtimes_* C_2$ with a unitary action $*$ (see Example \ref{ex-semidirect}).
By Proposition \ref{prop-semidirect}, $M_{L_2}$ is a submonoid of $C_3^{C_2}\rtimes C_2$, and also a submonoid of $\mathcal{T}_3^{C_2}\rtimes C_2$.\bqed
\end{example}

The following is a corollary of Theorem \ref{th-main} for $n=1$ and $\Gamma_1=\Sigma$.
\begin{corollary}\label{cor-main1}
Let $L\subseteq \Sigma^*$ be a regular language, and $M_L$ and $\eta_L$ be the syntactic monoid and the syntactic morphism of $L$. The following conditions are equivalent:\begin{enumerate}[(i)]
\item $L$ has a period $P>1$ with respect to $\Sigma$.
\item $M_L$ is a submonoid of $\TK^{C_P}\rtimes C_P$ where $K=\max\{|\eta_L(\Sigma^r(\Sigma^P)^*)|\mid r\in C_P\}$. Furthermore, $\eta_L(w)\in \TK^{C_P}\times \{\rho(w)\}$ for each $w\in\Sigma^*$.\bqed
\end{enumerate}
\end{corollary}
\begin{example}
As mentioned in Example \ref{ex-language}-(3), $L_3$ has period $2$ with respect to $\Sigma=\{a,b\}$. Therefore, $M_{L_3}$ is a submonoid of $\mathcal{T}_3^{C_2}\rtimes C_2$ and $\eta_{L_3}(\Sigma)\subseteq \mathcal{T}_3^{C_2}\times \{1\}$.
We can show that $M_{L_3}$ is a submonoid of $\overline{U}_2\times C_2$ by the mapping $\eta_{L_3}(a)\mapsto (\iota_1,1)$ and $\eta_{L_3}(b)\mapsto (\iota_2,1)$.
Furthermore, $\overline{U}_2\times C_2$ is a submonoid of $\overline{U}_2^{C_2}\rtimes C_2$ by Proposition \ref{prop-semidirect}, and also a submonoid of $\mathcal{T}_3^{C_2}\rtimes C_2$.\bqed
\end{example}

\subsection{Residual Monoids}
In this subsection, we focus on periods with respect to a fixed alphabet $\Sigma$. Let $L\subseteq\Sigma^*$ be a regular language that has a period $P$ with respect to $\Sigma$.
Our goal is to extract a monoid $T_r$ that recognizes $L_w=\{u\in (\Sigma^P)^*\mid wu\in L\}$ with $w\in \Sigma^r$ for each $0\le r<P$, where $T_r$ is designed to treat a word of length $P$ as a single letter.
Intuitively, $L_w$ is the `{\it periodic image}' of $L$ with residual $r=\rho(w)$.
For each $0\le r<P$, we define the subset $T_r\subseteq\TK$ as \[T_r=\{\tau\in \TK \mid \tau=f(r)\text{ with }(f,0)\in M_L\}~.\]
Note that $(f,0)$ is the abbreviation of $\Can^{-1}(f,0)$ for the canonical homomorphism $\Can:M_L\to \TK^{C_P}\rtimes C_P$ (see the first paragraph of Section 2.1).
We call $T_r$ the {\it residual monoid} (with residual $r$), and this definition is justified by the following fact.
\begin{lemma}
Let $L\subseteq \Sigma^*$ be a regular language that has a period $P$ with respect to $\Sigma$.
For each $0\le r<P$, the residual monoid $T_r$ forms a monoid.
\end{lemma}
\begin{proof}
  Because $T_r$ is a subset of the monoid $\TK$, it suffices to show that the monoid generated by $T_r$ is equal to $T_r$ itself.
  Therefore, we show that (i) $T_r$ has the identity element of $\TK$, and (ii) $T_r$ is closed under the composition $\circ$.

  (i): Let $M_L$ be the syntactic monoid of $L$, and let $\Can:M_L\to \TK^G\rtimes G$ be the canonical homomorphism.
  For the identity element $e$ of $M_L$, $\Can(e)=(f_e,0)$ satisfies that \[f_e(r)(k)=\begin{cases}
  {\theta_{{r}+0}}^{-1}(\theta_{r}(k)\cdot e)=\theta_r^{-1}(\theta_r(k))=k&\mbox{if~~~}k<|N_{r}|~,\\
  k&\mbox{if~~~}k\ge |N_{r}|
  \end{cases}
  \]
  for each $0\le r<P$. Therefore, $f_e(r)(k)=k$ for each $k\in K$, and $f_e(r)$ is the identity map on $K$.
  Because $\Can(e)=(f_e,0)\in M_L$, $T_r$ has the identity element $f_e(r)$.

  (ii): Let $f_1(r),f_2(r)\in T_r$ with $(f_1,0),(f_2,0)\in M_L$. Because\[(f_1,0)\cdot (f_2,0)=(f_1\cdot (0*f_2),0+0)=(f_1\cdot f_2,0),\] $(f_1\cdot f_2,0)\in M_L$ and $(f_1\cdot f_2)(r)\in T_r$. By the definition of the operation of the monoid $\TK^{C_P}$, $(f_1\cdot f_2)(r)=f_2(r)\circ f_1(r)\in T_r$.\qed
\end{proof}
Let $L\subseteq \Sigma^*$ be a regular language that has a period $P$ with respect to $\Sigma$.
Note that \begin{equation*}
L=\bigsqcup_{w\in \Sigma^{<P}}L\cap w(\Sigma^P)^*
\end{equation*}
where $\Sigma^{<P}=\bigcup_{0\le i<P}\Sigma^i$. That is, $L$ can be partitioned into $L\cap w(\Sigma^P)^*$ for each $w\in \Sigma^{<P}$.
We can show that each language $L_w=\{u\in (\Sigma^P)^*\mid wu\in L\}\subseteq (\Sigma^P)^*$ is fully recognized by the residual monoid $T_{\rho(w)}$.

\begin{theorem}\label{th-residualMonoid}
  Let $L\subseteq \Sigma^*$ be a regular language that has a period $P$ with respect to $\Sigma$.
  For each $w\in \Sigma^{<P}$, let $L_w=\{u\in (\Sigma^P)^*\mid wu\in L\}$ be the language over $\Sigma^P$.
  Then, the residual monoid $T_r$ fully recognizes $L_w$ where $r=\rho(w)$.
\end{theorem}
\begin{proof}
  We show that there exists a surjective monoid homomorphism $\eta_w:(\Sigma^P)^*\to T_r$ such that $L_w=\eta_w^{-1}(S)$ for some $S\subseteq T_r$. Let $M_L$ and $\eta_L$ be the syntactic monoid and syntactic morphism of $L$, and let $\Can:M_L\to \TK^{C_P}\rtimes C_P$ be the canonical homomorphism.

  We define the homomorphism $\eta_w$ as $\eta_w(u)=f(r)$ with $\Can(\eta_L(u))=(f,0)$ for each $u\in (\Sigma^P)^*$. By the definition of $T_r$, $\eta_w$ is surjective. Next, we define $S=\{\tau\in T_r\mid \theta_r(\tau(\theta_r^{-1}(\eta_L(w))))\in \eta_L(L)\}$.
  We show that $L_w=\eta_w^{-1}(S)$.
  Let $u\in (\Sigma^P)^*$ and $\tau=\eta_w(u)$.
  Then, \begin{eqnarray*}
  \tau(\theta_r^{-1}(\eta_L(w)))&=&\eta_w(u)(\theta_r^{-1}(\eta_L(w)))\\
  &=&f(r)(\theta_r^{-1}(\eta_L(w)))\\
  &=&{\theta_{r+0}}^{-1}(\theta_r(\theta_r^{-1}(\eta_L(w)))\cdot \eta_L(u))\\
  &=&\theta_r^{-1}(\eta_L(w)\cdot\eta_L(u))\\
  &=&\theta_r^{-1}(\eta_L(wu))
\end{eqnarray*}
and therefore $\theta_r(\tau(\theta_r^{-1}(\eta_L(w))))=\eta_L(wu)$.
Thus,\begin{align*}
u\in \eta_w^{-1}(S)&\iff \tau\in S&\\
&\iff \theta_r(\tau(\theta_r^{-1}(\eta_L(w))))\in \eta_L(L)&\\
&\iff \eta_L(wu)\in \eta_L(L)&\\
&\iff wu\in L &\mbox{(by $\eta^{-1}_L(\eta_L(L))=L$)}\\
&\iff u\in L_w~,&
\end{align*}
that is, $L_w=\eta_w^{-1}(S)$.\qed
\end{proof}

For a regular language $L\subseteq \Sigma^*$ that has a period $P$ with respect to $\Sigma$, we let $L_w$ and $\eta_w$ denote the language and the homomorphism mentioned in Theorem~\ref{th-residualMonoid} for each $w\in \Sigma^{<P}$.
Also, let $M_{L_w}$ and $\eta_{L_w}$ be the syntactic monoid and the syntactic morphism of $L_w\subseteq (\Sigma^P)^*$.
By Theorem~\ref{th-residualMonoid} and Proposition~\ref{prop-recognize}, there is a surjective homomorphism $\psi:T_{\rho(w)}\to M_{L_w}$.
{The following commutative diagram illustrates the relationship among the above monoids}:
\begin{equation}
  \xymatrix{
  {(\Sigma^P)^*}\ar@{->>}[r]^-{\eta_w}\ar[rd]_-{\eta_{L_w}}\ar@{}<2.0ex>\ar@{->>}[rd]|{}&{T_{\rho(w)}}\ar@{->>}[d]^-{\psi}\\
  &{M_{L_w}~.}
  }\label{eq:diagram}
\end{equation}

\section{Probabilities of Regular Languages}
For a language $L\subseteq \Sigma^*$, the {\it probability} (or {\it density}) of $L$ for length $\ell$, denoted by $\mu_L(\ell)$, is defined as $|L\cap \Sigma^\ell |/|\Sigma^\ell|$.
Note that $0\le \mu_L(\ell)\le 1$ for every $L\subseteq \Sigma^*$ and $\ell\ge 0$.
The probability of $L$, denoted by $\mu_L$, is the behavior of $\lim_{\ell\to \infty}\mu_L(\ell)$. The probabilities of languages have been extensively studied. In particular, the following proposition was proved in the study of formal power series \cite{powerseries}.
\begin{proposition}\label{prop-accumulation}
  The probability $\mu_L$ of every regular language $L\subseteq \Sigma^*$ has only finitely many accumulation points. Precisely, there exists a sequence of numbers $(\mu_0,\ldots,\mu_{P-1})$ with some $P\ge 1$ such that $\lim_{\ell\to \infty}\mu_L(\ell\cdot P+r)=\mu_r$ for each $0\le r<P$.\bqed
\end{proposition}
For a regular language $L\subseteq \Sigma^*$ that has $P$ accumulation points, we write $\mu_L=(\mu_0,\ldots,\mu_{P-1})$.
If $P=1$, we simply write $\mu_L=\mu_0$.
For example, the probability of $L_3=a(\Sigma\Sigma)^*\cup b\Sigma^*$  (see also Example~\ref{ex-language}-(3)) can be represented as $\mu_{L_3}=(\tfrac{1}{2},1)$ with two accumulation points.

\subsection{Periods and Markov Chains}
The probability of a regular language can be computed by a finite state Markov chain.
For a regular language $L\subseteq\Sigma^*$, let $\mathcal{M}$ be the finite state Markov chain (with uniform transition probabilities) defined as:\begin{itemize}
\item the state space is the set $Q$ of states of a DFA $\mathcal{A}$ that recognizes $L$, and
\item the transition matrix $\Pi\in [0,1]^{Q\times Q}$ is defined as $\Pi(q_i,q_j)=|\{a\in \Sigma \mid q_i\xrightarrow{a}q_j$ is a transition of $\mathcal{A}\}|/|\Sigma|$ for each $q_i,q_j\in Q$.
\end{itemize}
Then, it is clear that the probability $\mu_L(\ell)$ is equal to $\sum_{q\in F}\Pi^\ell (q_0,q)$ where $q_0$ and $F$ are the initial state and the set of final states of $\mathcal{A}$, respectively.
\begin{example}
Let $\mathcal{A}_3$ be the DFA that recognizes $L_3=a(\Sigma\Sigma)^*\cup b\Sigma^*$ defined in Example~\ref{ex-language}-(3). Then, the corresponding Markov chain $\mathcal{M}_3$ is as shown in Figure~4 and the transition matrix is \[
\Pi=\begin{pmatrix}
0&\tfrac{1}{2}&0&\tfrac{1}{2}\\
0&0&1&0\\
0&1&0&0\\
0&0&0&1\end{pmatrix}~.
\]
The probability of $L_3$ for length $\ell$ is equal to $\Pi^{\ell}(1,2)+\Pi^{\ell}(1,4)$ because $q_1$ and $\{q_2,q_4\}$ are the initial state and the set of final states of $\mathcal{A}_3$, respectively.
For example, $\mu_{L_3}(2)=\Pi^2(1,2)+\Pi^2(1,4)=1/2$. In fact, $|L_3\cap \Sigma^2|/|\Sigma^2|=|\{ba,bb\}|/4=1/2$.\bqed
\end{example}
{\centering
\begin{tikzpicture}[auto]
  \node [scale=0.65] (s0) [state] at (0,0) {\LARGE $q_3$};
  \node [scale=0.65] (s1) [state] at (1.4,0) {\LARGE $q_1$};
  \node [scale=0.65] (s2) [state] at (0,-1.4) {\LARGE $q_2$};
  \node [scale=0.65] (s3) [state] at (2.8,-1.4) {\LARGE $q_4$};

  \path [-stealth, thick]
    (s3) edge [loop above] node {$1$} ()
    (s1) edge [bend right] node [below] {$\tfrac{1}{2}$} (s3)
    (s1) edge [bend left] node [below] {$\tfrac{1}{2}$} (s2)
    (s0) edge node {$1$} (s2)
    (s2) edge [bend left] node {$1$} (s0);
\end{tikzpicture}
\\\vspace{5pt}\hspace{112pt}{\bf Figure 4.} Markov chain $\mathcal{M}_3$
}

\vspace{10pt}

A strongly connected component of a directed graph is said to be a {\it sink} if there are no outgoing edges from itself.
A sink $Q'$ of a Markov chain is said to be {\it $P$-periodic} if $P$ is the maximum number satisfying: $\ell$ is a multiple of $P$ if $\Pi^{\ell}(q,q)>0$ with some $q\in Q'$.
For example, Markov chain $\mathcal{M}_3$ (see Figure~4) has two sinks: $2$-periodic sink $\{q_2,q_3\}$ and $1$-periodic sink $\{q_4\}$.
Note that for a DFA $\mathcal{A}$ and the corresponding Markov chain $\mathcal{M}$, $\Pi^{\ell}(q,q)>0$ iff there is a run $q\xrightarrow{w}q$ for some $w\in \Sigma^\ell$.
Therefore, we also say that a sink $Q'$ of a DFA is $P$-periodic if $P$ is the maximum number satisfying: for each  $w\in\Sigma^*$, $|w|$ is a multiple of $P$ if $q\xrightarrow{w}q$ with some $q\in Q'$.
$P$-periodicity affects the limit distribution of the Markov chain. In fact, we can show that for the transition matrix $\Pi$ of every finite state Markov chain, the limit distribution $\lim_{\ell\to \infty}\Pi^\ell$ oscillates around $P$ accumulation points where $P$ is a divisor of the least common multiple of the periods of all sinks.
By the correspondence between the probability of a regular language and the limit distribution of the Markov chain, we have the following property.
\begin{proposition}\label{prop-sinkPeriod}
  Let $L\subseteq \Sigma^*$ be a regular language, and let $\mathcal{A}$ be a DFA that recognizes $L$.
  The probability of $L$ can be represented as $\mu_L=(\mu_0,\ldots,\mu_{P-1})$ with $\mu_0,\ldots,\mu_{P-1}\in [0,1]$ where $P$ is the least common multiple of
  $\{P'\in\mathbb{N}\mid Q'$ is $P'$-periodic where $Q'$ is a sink of $\mathcal{A}\}$.\bqed
\end{proposition}
As can be seen from the definition, $P$-periodicities of sinks correspond to the periods of a language with respect to $\Sigma$ defined in Section 2.3.
Actually, our definition of period is inspired by the studies of Markov chains.
The next lemma describes this correspondence in detail.
\begin{lemma}\label{lem-syntPeriod}
Let $L\subseteq \Sigma^*$ be a regular language, and $M_L$ be the syntactic monoid of $L$.
If a $P$-periodic sink exists in the Cayley graph of $M_L$, then the maximum period of $L$ with respect to $\Sigma$ is a multiple of $P$.
\end{lemma}
\begin{proof}
  Let $Q'$ be a $P$-periodic sink of the graph.
  We assume that a word $w\in\Sigma^*$ satisfies $t\cdot \eta_L(w)=t$ for some $t\in M_L$ where $\eta_L$ is the syntactic morphism of $L$.
  By the definition of the period of a language, it suffices to show that $|w|$ is always a multiple of $P$.
  Suppose for contradiction that $|w|$ is not a multiple of $P$.
  Let $s$ be an arbitrary element of $Q'$.
  Because $Q'$ is a sink,
  $s\cdot t\in Q'$ and
  there is $s'\in M_L$ such that $s\cdot t\cdot s'=s$.
  Note that $s\cdot t\cdot \eta_L(w)\cdot s'=s\cdot t\cdot s'=s$ by the assumption $ t\cdot \eta_L(w)=t$.
  Therefore, for any $u\in \eta_L^{-1}(t)$ and $v'\in \eta_L^{-1}(s')$, it holds that $s\cdot \eta_L(u\cdot w^{i}\cdot v')=s$ for each $i\ge 0$.
  As we supposed that $|w|$ is not a multiple of $P$, there is $j\ge 1$ such that the length of $w'=uw^jv'$ is not a multiple of $P$.
  However, this contradicts the assumption that $Q'$ is $P$-periodic because $s\cdot \eta_L(w')=s$ and $|w'|$ is not a multiple of $P$.\qed
\end{proof}

This proof implies that for any sinks $Q',Q''\subseteq M_L$, if $Q'$ is $P'$-periodic, then $Q''$ is $P''$-periodic where $P''$ is a multiple of $P'$.
Therefore, the periodicity of each sink of the syntactic monoid is equal to each other.
For example, the Cayley graph of $M_{L_3}$ in Example~\ref{ex-language} has two sinks, and both are $2$-periodic.
For this reason, the Cayley graph of the syntactic monoid is convenient for discussing the periodicities and probabilities compared to the minimal DFA or other DFAs.
Moreover, by Lemma~\ref{lem-disjoint}, we see that the subsets $N_0,\ldots,N_{P-1}\subseteq M_L$ defined in Section~3 naturally correspond to the {\it periodic classes} in the context of Markov chains.

Note that the Cayley graph of the syntactic monoid $M_L$ of a regular language $L$ can be regarded as a DFA that recognizes $L$. By Proposition~\ref{prop-sinkPeriod} and Lemma~\ref{lem-syntPeriod}, we have the following fact.
\begin{theorem}\label{th-maxperiod}
  Let $L\subseteq \Sigma^*$ be a regular language, and $P$ be the maximum period of $L$ with respect to $\Sigma$. Then, the probability of $L$ can be represented as $\mu_L=(\mu_0,\ldots,\mu_{P-1})$ with $\mu_0,\ldots,\mu_{P-1}\in [0,1]$.\bqed
\end{theorem}
{
The number of accumulation points of the probability can be smaller than the maximum period $P$. For example, the maximum period of the language $L=a(\Sigma\Sigma)^*\cup b\Sigma(\Sigma\Sigma)^*\subseteq\{a,b\}^*$ with respect to $\Sigma$ is $2$, and $\mu_L=(\tfrac{1}{2},\tfrac{1}{2})$ holds. However, $\tfrac{1}{2}$ is the only accumulation point, and therefore, $\mu_L$ can also be represented as $\mu_L=\tfrac{1}{2}$.}

\subsection{Applications}
In this section, we give an application of the algebraic decomposition discussed in Section~3.

A monoid $M$ is said to be {\it zero} if there exists the zero element $\iota\in M$, i.e., $\iota\cdot m=\iota=m\cdot \iota$ for each $m\in M$.
Note that every monoid $M$ can have at most one zero element.
For a subset $I$ of a monoid $M$, $I$ is an {\it ideal} of $M$ if $IM=I=MI$.
For example, the singleton of the zero element $\{\iota\}\subseteq M$ is an ideal of every zero monoid $M$.
The following proposition can be easily shown (e.g., see Chapter II-3 of \cite{Pin}).
\begin{proposition}\label{prop-ideals}
Let $M$ and $N$ be monoids, and $\psi:M\to N$ be a surjective monoid homomorphism. Then, the followings hold.
\begin{enumerate}[(i)]
\item If $I$ is an ideal of $M$, then $\psi(I)$ is an ideal of $N$. In particular, if $\iota$ is the zero element of $M$, then $\psi(\iota)$ is the zero element of $N$.
\item If $J$ is an ideal of $N$, then $\psi^{-1}(J)$ is an ideal of $M$. In particular, if $\iota$ is the zero element of $N$, then $\psi^{-1}(\iota)$ is an ideal of $M$.
\item Let $I$ be a non-empty ideal of $M$. Then, $M'=M\setminus I\cup \{\iota\}$ forms a monoid with the zero element $\iota$ where the operation $\times$ of $M'$ is defined as\[
s_1\times s_2=\begin{cases}
s_1\cdot s_2&\text{if }s_1,s_2,s_1\cdot s_2\in M\setminus I\\
\iota&\text{otherwise}
\end{cases}
\]for each $s_1,s_2\in M'$. \bqed
\end{enumerate}
\end{proposition}
The monoid $M'$ in Proposition~\ref{prop-ideals}-(iii) is called the {\it Rees factor monoid} of $M$ modulo $I$, which originates from the study of semigroups by Rees (\cite{Rees}).

The following interesting fact is known as a characterization of zero syntactic monoid.
\begin{proposition}[\cite{zero-one}]\label{prop-zeroone}
  Let $L\subseteq \Sigma^*$ be a regular language. The following conditions are equivalent:
\begin{enumerate}[(i)]
  \item $\mu_L=0$ or $\mu_L=1$.
  \item The syntactic monoid of $L$ is zero.\bqed
\end{enumerate}
\end{proposition}\vspace{5pt}
This proposition implies that if the probability of a regular language $L$ has only one accumulation point $\mu_0$, then we can determine whether $\mu_0\in\{0,1\}$ by examining the existence of the zero element in the syntactic monoid.
We would like to generalize this result to determine whether $\mu_r\in\{0,1\}$ for arbitrary $0\le r<P$ when $\mu_L$ is represented as $(\mu_0,\ldots,\mu_{P-1})$.
Remember that $L_w=\{u\in (\Sigma^P)^*\mid wu\in L\}$ is a language over $\Sigma^P$.
Because \[
\mu_r=\sum_{w\in \Sigma^r}\frac{1}{|\Sigma|^r}\mu_{L_w}\] holds by a simple discussion on probabilities,
we have $\mu_r=1$ iff $\mu_{L_w}=1$ for each $w\in\Sigma^r$, and $\mu_r=0$ iff $\mu_{L_w}=0$ for each $w\in\Sigma^r$.
Therefore, a naive way to determine whether $\mu_r\in\{0,1\}$ is to construct the syntactic monoid of $L_w\subseteq (\Sigma^P)^*$ and examine the existence of the zero element for each $w\in\Sigma^r$.
Nevertheless, we already obtain the monoid $T_r$ that fully recognizes $L_w$ (see Theorem~\ref{th-residualMonoid}), and thus we can directly determine whether $\mu_r\in\{0,1\}$ without constructing the syntactic monoid of each $L_w$.
Note that $r = \rho(w)$ is the residual of $w$ and hence we can uniformly use the same $T_r$ for all $w$ that have the same residual.
{By the diagram \eqref{eq:diagram} and Proposition~\ref{prop-zeroone}, we can show the following Theorem~\ref{th-ex-zeroone}. Intuitively, the $r$-th accumulation point of the probability is either $0$ or $1$ iff there exists an ideal consisting only of either rejecting states or accepting states in the $r$-th periodic class of the syntactic monoid.}

\begin{theorem}\label{th-ex-zeroone}
  Let $L\subseteq \Sigma^*$ be a regular language, and let $P$ be the maximum period with respect to $\Sigma$. For each $w\in \Sigma^r$ with $0\le r<P$, the limit $\mu_{L_w}=\lim_{\ell\to\infty}\mu_{L_w}(\ell)$ exists and the following conditions are equivalent:
\begin{enumerate}[(i)]
  \item $\mu_{L_w}=0$ or $\mu_{L_w}=1$.
  \item There exists a non-empty ideal $I$ of $T_r$ such that either $I\cap \eta_w(L_w)=\emptyset$ or $I\subseteq \eta_w(L_w)$ holds where $\eta_w$ is the homomorphism in the diagram \eqref{eq:diagram}.
\end{enumerate}
\end{theorem}
\begin{proof}
  If $\lim_{\ell\to\infty}\mu_{L_w}(\ell)$ oscillates for some $w\in\Sigma^r$, $\mu_L$ cannot be represented as $(\mu_0,\ldots,\mu_{P-1})$ and this contradicts Theorem~\ref{th-maxperiod}.
  Therefore, $\mu_{L_w}=\lim_{\ell\to\infty}\mu_{L_w}(\ell)$ exists. We show (i)$\Leftrightarrow$(ii).

  (i)$\Rightarrow$(ii): Let $M_{L_w}$ be the syntactic monoid of $L_w$ and let $\iota$ be the zero element of $M_{L_w}$, of which existence is guaranteed by Proposition~\ref{prop-zeroone}.
  Let $I=\psi^{-1}(\iota)$ where $\psi$ is the homomorphism in the diagram \eqref{eq:diagram}. Then, $I$ is an ideal of $T_r$ by Proposition~\mbox{\ref{prop-ideals}-(ii)}.
  Now, if $\iota\in \psi\circ \eta_w(L_w)$, then $I=\psi^{-1}(\iota)\subseteq \eta_w(L_w)$.
  If $\iota\notin \psi\circ \eta_w(L_w)$, then $\{\iota\}\cap \psi\circ \eta_w(L_w)=\emptyset$ and $I\cap \eta_w(L_w)=\emptyset$.

  (ii)$\Rightarrow$(i): Let $I$ be a non-empty ideal of $T_r$, and let $T_r'=T_r\setminus I\cup\{\iota\}$ be the Rees factor monoid modulo $I$ (see Proposition~\ref{prop-ideals}-(iii)).
  It is clear that the mapping $\phi:T_r\to T_r'$ defined as \[
    \phi(s)=\begin{cases}
    s&\text{if }s\in T_r\setminus I,\\
    \iota&\text{if }s\in I
  \end{cases}
    \]is a surjective homomorphism.
    For $S=\eta_w(L_w)$, it holds that $L_w=(\phi\circ \eta_w)^{-1}(S\setminus I\cup\{\iota\})$ if $I\subseteq S$, and $L_w=(\phi\circ \eta_w)^{-1}(S\setminus I)=(\phi\circ \eta_w)^{-1}(S)$ if $I\cap S=\emptyset$.
    Therefore, $T_r'$ also fully recognizes $L_w$ with $\phi\circ \eta_w$.
    By Proposition~\ref{prop-recognize}, there is a surjective homomorphism $\psi'$ from $T_r'$ into the syntactic monoid $M_{L_w}$.
    Because $T_r'$ is zero, $M_{L_w}$ is also zero by Proposition~\ref{prop-ideals}-(i). This implies that $\mu_{L_w}\in\{0,1\}$ by Proposition~\ref{prop-zeroone}.\qed
\end{proof}

\begin{example}
Let $L_3=a(\Sigma\Sigma)^*\cup b\Sigma^*$ (see also Example~\ref{ex-language}-(3)). Note that the maximum period with respect to $\Sigma$ is $2$, and the probability of $L_3$ is $\mu_{L_3}=(\tfrac{1}{2},1)$.
For any $w \in \Sigma^*$, we abbreviate $(L_3)_w$ as $L_w$ in the following.
We have $\mu_{L_\varepsilon}\notin \{0,1\}$ where $L_\varepsilon=ba(\Sigma\Sigma)^*\cup bb(\Sigma\Sigma)^*$, and $\mu_{L_a}=\mu_{L_b}=1$ where $L_a=L_b=(\Sigma\Sigma)^*$.
In fact, $T_1$ is the trivial group $\{e\}$, and $\{e\}$ itself is an ideal satisfying $\{e\}=\eta_a(L_a)=\eta_b(L_b)$.
On the other hand, $T_0=\overline{U}_2=\{e,\iota_1,\iota_2\}$ and $\eta_{L_\varepsilon}(L_\varepsilon)$ contains only one of either $\iota_1$ or $\iota_2$.
Because the only non-empty ideals of $\overline{U}_2$ are $\overline{U}_2$ itself and $\{\iota_1,\iota_2\}$, the condition (ii) in Theorem~\ref{th-ex-zeroone} cannot be satisfied.
Intuitively, the three states in the top row of Figure 2--3 correspond to the periodic class $T_0$, while the two states in the bottom row correspond to the periodic class $T_1$.\bqed
\end{example}

\section{Krohn-Rhodes Decompositions}
We get back to general periods, including the cases that $\Gamma\subsetneq \Sigma$. In Section~3, we showed that every language with periods can be decomposed into a semidirect product with cyclic groups. In this section, we show that every language with periods also can be decomposed into another kind of product, called the wreath product, with cyclic groups.

For a right action $*:X\times M\to X$, we say that $(X,M)$ is a {\it transformation monoid}\footnote{This is a conventional term. A transformation monoid is the pair of an action and a monoid, but not a monoid itself. Nevertheless, similarly to the existence of a monoid simulating a DFA, there exists a monoid `isomorphic' to each transformation monoid.}. For example, $(M,M)$ is a transformation monoid for every monoid $M$ with the monoid operation $*$ of $M$. A transformation monoid is said to be finite if both of $X$ and $M$ are finite.
For two transformation monoids $(X,M)$ and $(Y,N)$, a pair $(\phi,\psi)$ of $\phi:X\to Y$ and $\psi:M\to N$ is a {\it homomorphism} if (i) $\psi$ is a monoid homomorphism, and (ii) $\phi(x*m)=\phi(x)*\psi(m)$ for each $x\in X,m\in M$.
We say that $(\phi,\psi)$ is surjective if both of $\phi$ and $\psi$ are surjective.
If there are a subset $X'$ of $X$ and a submonoid $M'$ of $M$ with a surjective homomorphism $(X',M')$ to $(Y,N)$, then $(Y,N)$ is called a {\it divisor} of $(X,M)$.

For two transformation monoids $(X,M)$ and $(Y,N)$, the transformation $(X\times Y,M^Y\rtimes N)$ where the action $*:(X\times Y)\times (M^Y\rtimes N)\to (X\times Y)$ is defined as\begin{equation}
(x,y)*(f,n)=(x* f(y),y*n) \label{eq:wreathAction}
\end{equation}
is called the {\it wreath product} of $(X,M)$ and $(Y,N)$, and denoted by $(X,M)\wr(Y,N)$. The following property is well known.
\begin{proposition}
  Let $(X_1,M_1),(X_2,M_2),(X_3,M_3)$ be transformation monoids. Then, $((X_1,M_1) \wr (X_2,M_2)) \wr (X_3,M_3)$ is isomorphic to $(X_1,M_1)\wr ((X_2,M_2)\wr (X_3,M_3))$.\bqed
\end{proposition}

Krohn and Rhodes showed the {\it decomposition theorem} of finite monoids (\cite{Krohn-Rhodes}).

\begin{proposition}[Krohn-Rhodes]\label{prop-Krohn}
  Every finite transformation monoid $(X,M)$ is a divisor of a wreath product of the form\begin{equation*}
  (Y_1,N_1)\wr\cdots\wr(Y_k,N_k)
\end{equation*}
  where each $(Y_i,N_i)$ with $1\le i\le k$ is $(U_2,U_2)$ or $(G,G)$ with a nontrivial group $G$ dividing $M$.\bqed
\end{proposition}
Because every DFA can be regarded as a transformation monoid, Krohn-Rhodes theorem is also referred to as the decomposition theorem of regular languages.
Nevertheless, known proofs for Proposition~\ref{prop-Krohn} are purely semigroup-theoretic (e.g., \cite{localDivisor,q-theory}), and there is a gap between the proofs and formal language theory.
Moreover, a finite monoid is not always the syntactic monoid of a language. For example, it is explained in \cite{Lawson} that $U_K$ with $K\ge 3$ is not a syntactic monoid (see Example~\ref{ex-basic_monoids} for the definition of $U_K$).
For these reasons, there is a need for studies on Krohn-Rhodes decompositions for syntactic monoids.
As an application of the results of this paper, we give a partial Krohn-Rhodes decomposition of syntactic monoids of periodic regular languages.
More precisely, for any periodic regular language $L$ that has periods $P_1, \ldots, P_n$, $(M_L, M_L)$ is a divisor of $(\TK, \TK) \wr (G, G)$ for some $K$ where $G$ is the direct product of the cyclic groups $C_{P_i}$ for $1\le i\le n$.

\begin{theorem}\label{th-main2}
  Let $L\subseteq \Sigma^*$ be a regular language and $M_L$ be the syntactic monoid of $L$. Let $L$ have periods $P_1,\ldots,P_n> 1$ with respect to $\Gamma_1,\ldots,\Gamma_n\subseteq \Sigma$, respectively. Then, $(M_L,M_L)$ is a divisor of the wreath product of the form\[
  (\TK,\TK)\wr (G,G)
  \] where $G=C_{P_1}\times \cdots\times C_{P_n}$ and $K=\max\{|N_{\vb{r}}|\mid \vb{r}\in G\}$. Furthermore, $G$ divides $M_L$.
\end{theorem}
\begin{proof}
  Because $\overline{\rho}:M_L\to G$ is a surjective homomorphism, $G$ is a divisor of $M_L$. It suffices to show that there exists a surjective homomorphism $(\phi,\psi)$ from $(\TK,\TK)\wr (G,G)=(\TK\times G,\TK^G\rtimes G)$ to $M_L$.

  Let $\Can:M_L\to \TK^G\rtimes G$ be the canonical homomorphism. Because $\Can$ is an injective homomorphism, the inverse $\psi$ of $\Can$ is a partial surjective homomorphism from $\TK^G\rtimes G$ to $M_L$. Let $\hat\phi(t)=(f(\vb{0}),\overline{\rho}(t))$ for $\Can(t)=(f,\overline{\rho}(t))$,
  then $\hat\phi$ is an injective homomorphism from $M_L$ to $\TK\times G$ by the same reason as the proof of Theorem~\ref{th-main}. Therefore, the inverse $\phi$ of $\hat\phi$ is a partial surjective mapping from $\TK\times G$ to $M_L$.
  Specifically, \begin{equation}\phi(f(\vb{0}),\vb{c})={\theta_{\vb{c}}}(f(\vb{0})({\theta_{\vb{0}}}^{-1}(e)))\label{eq:phi}
  \end{equation} for each $(f(\vb{0}),\vb{c})\in \TK\times G$.

  We show that $\phi(x* m)=\phi(x)\cdot\psi(m)$ for each $x\in \TK\times G$ and $m\in \TK^G\rtimes G$. Let $x=(f(\vb{0}),\vb{c})$ and $m=(g,\vb{r})$.
  By the definition of $\Can$, it holds that $g(\vb{c})(k)={\theta_{\vb{c}+\vb{r}}}^{-1}(\theta_{\vb{c}}(k)\cdot \psi(m))$ for each $0\le k<K$, and hence,
  \begin{equation} \theta_{\vb{c}+\vb{r}}(g(\vb{c})(k))=\theta_{\vb{c}}(k)\cdot \psi(m)\label{eq:transformationAction}
  \end{equation}
  holds. Then, \begin{align*}
    \phi(x)\cdot\psi(m)&=\theta_{\vb{c}}(f(\vb{0})({\theta_0}^{-1}(e)))\cdot \psi(m)\\
    &={\theta_{{\vb{c}}+\vb{r}}}(g(\vb{c})(f(\vb{0})({\theta_0}^{-1}(e))))&\tag{\text{by \eqref{eq:transformationAction}}}\\
    &=\phi(g(\vb{c})\circ f(\vb{0}),\vb{c}+\vb{r})\tag{\text{by \eqref{eq:phi}}}\\
    &=\phi((f(\vb{0}),{\vb{c}})*(g,\vb{r})) \tag{\text{by \eqref{eq:wreathAction}}}\\
    &=\phi(x* m)~.
  \end{align*}
  Therefore $(M_L,M_L)$ is a divisor of $(\TK,\TK)\wr (G,G)$.\qed
\end{proof}
{Note that this theorem provides only a partial decomposition: the monoid $\TK$ is not represented as a wreath product of monoids of the form $(U_2,U_2)$ and groups in general.
For this partial decomposition to be useful, the structure embedded in $\TK$ needs to be simpler than $M_L$. This aspect is related to the group complexity, which was introduced in~\cite{groupComplexity}.
The smallest number of components $(Y_i,N_i)$ of the form $(G,G)$ in Proposition~\ref{prop-Krohn} over all possible decompositions is called the {\it group complexity} of $(X,M)$.
There are several studies for computing group complexity, but any complete algorithm has not yet been obtained (see e.g. \cite{q-theory} for detail).
Therefore, if the decomposition given in Theorem~\ref{th-main2} leads to the minimal decomposition, that is, if the minimal decomposition of $(M_L,M_L)$ can be represented as \[
(Y_1,N_1)\wr\cdots\wr(Y_k,N_k)\wr (G,G)
\] with $G=C_{P_1}\times\cdots\times C_{P_n}$, our result is beneficial for computing the group complexity. Proving this is left for future work, but it appears to be challenging.}

\section{Conclusions}
We have provided an algebraic decomposition of regular languages with periods using cyclic groups and semidirect products. Furthermore, we have discussed several applications of our results including probabilities, Markov chains, and Krohn-Rhodes decompositions.
{In Sections~3~and~5, we have already outlined the following future work:
\begin{itemize}
  \item In the decomposition given by Theorem~\ref{th-main}, is it possible to restrict the first component $\TK^{C_{P_1}\times \cdots \times C_{P_n}}$ to a simpler monoid?
  \item Does the partial decomposition given in Theorem~\ref{th-main2} lead to the minimal Krohn-Rhodes decomposition of the syntactic monoid?
\end{itemize}
We believe that these two problems are inherently related.}

In addition, further investigation is needed regarding the relation between our decomposition and the holonomy decomposition of automata \cite{Eilenberg,Cascade}. As mentioned in Section~1, the holonomy decomposition is known as a decomposition of automata, and is closely related to the Krohn-Rhodes decomposition. Therefore, we believe that there exists a connection between the decomposition of syntactic monoids we provided and the holonomy decomposition.

Finally, extending the definition of periods is also left for future work.
For example, each non-empty subset $\Gamma\subseteq \Sigma$ can be regarded as a {\it code} because every word in $\Gamma^*$ has a unique factorization in {\it codewords} in $\Gamma$ (see e.g., \cite{code} in detail).
Therefore, the definition of periods could be extended by the number of occurrences of codewords in $\Gamma$ with another code $\Gamma\subseteq \Sigma^*$.


\begin{thebibliography}{99}
  \bibitem{code} J. Berstel, D. Perrin and C. Reutenauer, Codes and automata, Encyclopedia of Mathematics and its Applications, Vol. 129, Cambridge University Press, 2010.

  \bibitem{Eilenberg} S. Eilenberg, Automata, Languages, and Machines, Academic Press, 1976.

  \bibitem{localDivisor} V. Diekert, M. Kufleitner and B. Steinberg, The Krohn-Rhodes Theorem and Local Divisors, Fundamenta Informaticae, 116(1-4), 65-77, 2012.

  \bibitem{groupComplexity} K. Krohn and J. L. Rhodes, Complexity of Finite Semigroups, Annals of Mathematics. Second Series, 88(1), 128--160, 1968.

  \bibitem{Krohn-Rhodes} K. Krohn, J. L. Rhodes and B. R. Tilson, Algebraic Theory of Machines, Languages, and Semigroups, 1, Academic Press, 81--125, 1968.

  \bibitem{Lallement} G. Lallement, Semigroups and Combinatorial Applications, John Wiley \& Sons, Inc., 1979.

  \bibitem{Lawson} M. V. Lawson, Finite Automata, Chapman and Hall/CRC, 2003.

  \bibitem{Norris} J. R. Norris, Markov Chains, Cambridge University Press, 1998.

  \bibitem{Pin} J. \'{E}. Pin, Mathematical Foundations of Automata Theory, Lecture notes LIAFA, Universit\'{e} Paris, 2010.

  \bibitem{Rees} D. Rees, On Semi-groups, Mathematical Proceedings of the Cambridge Philosophical Society, 36(4), Cambridge University Press, 1940.

  \bibitem{q-theory} J. L. Rhodes and B. Tilson, The $\mathfrak{q}$-theory of Finite Semigroups, Springer Science \& Business Media, 2009.

  \bibitem{powerseries} A. Salomaa and S. Matti, Automata-theoretic Aspects of Formal Power Series, Springer Science \& Business Media, 2012.

  \bibitem{Schutzenberger} M. P. Sch\"{u}tzenberger, On Finite Monoids Having only Trivial Subgroups, Information and Control, 8, 190--194, 1965.

  \bibitem{zero-one} R. Sin'ya, An Automata Theoretic Approach to the Zero-One Law for Regular Languages: Algorithmic and Logical Aspects, Electronic Proceedings in Theoretical Computer Science 193:172--185, 2015.

  \bibitem{LTmeasure} R. Sin'ya, Measuring Power of Locally Testable Languages, International Conference on Developments in Language Theory, 274--285, 2022.

  \bibitem{Cascade} H. P. Zeiger, Cascade Synthesis of Finite-state Machines, Information and Control, 10, 419--433, 1967.
\end{thebibliography}
\end{document}